\documentclass[10pt, conference, letterpaper]{IEEEtran}

\usepackage{framed}
\usepackage{hyperref}
\usepackage{tcolorbox}
\usepackage{cite}
\usepackage{amsmath}
\usepackage{amsfonts}
\usepackage{mathtools}
\usepackage{amssymb}
\usepackage{amsthm}
\usepackage[ruled,linesnumbered]{algorithm2e}
\usepackage{cases}
\usepackage{graphicx}

\newtheorem{theorem}{{\bf Theorem}}

\newtheorem{lemma}{{\bf Lemma}}

\newtheorem{corollary}{{\bf Corollary}}

\newtheorem*{ass:betatheta}{{\bf Assumption \ref{ass:betatheta}}}

\IEEEoverridecommandlockouts

\begin{document}
%
\title{Minimizing Age of Information in Spatially Distributed Random Access Wireless Networks}
\author{Nicholas Jones and Eytan Modiano\\
Laboratory for Information and Decision Systems, MIT\\
jonesn@mit.edu,  modiano@mit.edu
\thanks{This work was supported by NSF Grant CNS-1713725 and by Army Research Office (ARO) grant number W911NF-17-1-0508.}
\vspace{-3mm}
}

\IEEEaftertitletext{\vspace{-0.6\baselineskip}}
\maketitle

\begin{abstract}

We analyze Age of Information (AoI) in wireless networks where nodes use a spatially adaptive random access scheme to send status updates to a central base station. We show that the set of achievable AoI in this setting is convex, and design policies to minimize weighted sum, min-max, and proportionally fair AoI by setting transmission probabilities as a function of node locations. We show that under the capture model, when the spatial topology of the network is considered, AoI can be significantly improved, and we obtain tight performance bounds on weighted sum and min-max AoI. Finally, we design a policy where each node sets its transmission probability based only on its own distance from the base station, when it does not know the positions of other nodes, and show that it converges to the optimal proportionally fair policy as the size of the network goes to infinity.
\end{abstract}


%
\IEEEpeerreviewmaketitle

\section{Introduction}

Wireless networks emerged as a primary way to enable internet connectivity and communication between individuals. As these networks have become more sophisticated, their capacity has increased to allow higher levels of data throughput. However, with the rise of connected devices and the Internet of Things (IoT), the purpose and design of some networks is changing. In contrast to networks of human users, IoT networks generally have lower data rates but stricter latency requirements. Instead of data intensive uses like video streaming, these networks commonly involve small update packets which are time sensitive.

Age of Information (AoI) was introduced to provide a metric for information freshness in networks \cite{kaul2012real} and measures the time elapsed since the last received update from a source was generated. Significant work has been done in designing scheduling policies to minimize AoI, and the results in these areas are quite mature \cite{sun2017update, kadota2018scheduling, kadota2019scheduling, talak2020optimizing, talak2018distributed, tripathi2019whittle}.

All of the works above focus on centralized scheduling policies, and less work has been done to minimize AoI in random access channels, where nodes decide whether to access the channel in a completely decentralized way. In \cite{yates2017status}, the authors analyze AoI with slotted ALOHA and find the optimal attempt probabilities to minimize weighted sum AoI. In \cite{kadota2021age}, the authors analyze random access networks with stochastic arrivals and optimize the transmission probabilities for both slotted ALOHA and CSMA. In \cite{maatouk2020age}, the authors analyze AoI in CSMA and optimize the back-off timers to minimize Age when updates are generated at will. 

In \cite{chen2022age, yavascan2021analysis}, the authors propose an AoI threshold below which nodes will never access the channel. When their AoI exceeds the threshold, they participate in slotted ALOHA, and otherwise remain silent to prioritize nodes with larger AoI. In \cite{ahmetoglu2022mista}, a hybrid between ALOHA and CSMA is proposed, where nodes use an Age threshold coupled with carrier sensing.

All of these works assume a collision interference model, meaning that if multiple nodes access the channel simultaneously none of them succeed. When nodes are distributed in space, this is often an oversimplification. At the physical layer, wireless signals attenuate over distance, so transmissions from varying distances will be perceived differently at the receiver. When two nodes, one close by and one far away, transmit simultaneously, it is common for the farther node's transmission to be drowned out and for the closer node's transmission to succeed, contrary to what's predicted in the collision model. This leads to unfairness in the network, where farther nodes can experience much larger AoI on average. This is especially true in CSMA-based protocols like 802.11 \cite{bianchi2000performance}, where (i) sensing the channel may fail to detect farther nodes' transmissions, and (ii) successes cause nodes to transmit more frequently and failures cause them to transmit less frequently. This creates a positive feedback loop where closer nodes transmit more often, amplifying the spatial unfairness. When applied to human users, this phenomenon may not be catastrophic, but in an IoT network with critical time-sensitive updates, this unfairness can result in safety or reliability issues.

A commonly used model which accounts for spatial effects is the capture model, and some work has been done under this model analyzing random access networks in regards to throughput. In \cite{celik2009mac}, the authors highlight the spatial unfairness in CSMA-based protocols specifically with multipacket reception, and design a protocol which counter-intuitively backs off after a success, sharing a similar flavor to threshold ALOHA.

In \cite{baccelli2013adaptive}, the authors analyze what they call an adaptive spatial ALOHA protocol, which operates like slotted ALOHA but allows nodes to transmit with varying probabilities based on their spatial location. They optimize this protocol to maximize sum, max-min, and proportionally fair throughput, both with full topology information and without. They extend these results in \cite{baccelli2014analysis}, where they introduce the notion of a stopping set, or limited topology information. Following a similar line, \cite{yang2020optimizing} derives a spatial ALOHA policy to minimize sum AoI in a network of transmitter-receiver pairs using stochastic geometry and the same idea of stopping sets. In \cite{mankar2021spatial}, the authors find a spatial distribution for the mean Peak AoI in a network of source-destination pairs.

In this work, we analyze AoI in a network of sources sending status updates to a central base station. These sources use a random access policy, which sets transmission probabilities as a function of node locations. We show that the set of achievable AoI in this setting is convex, and derive policies to minimize weighted sum, min-max, and proportionally fair AoI. We show that under the capture model, when node locations are considered, AoI can be significantly improved, and the spatial unfairness of traditional random access (c.f. Figure \ref{fig:h_comp_vs_r}) can be reduced or eliminated. We derive tight performance bounds on weighted sum and min-max AoI, which also provide fairness guarantees. Finally, we design a topology-agnostic policy, where each node sets its transmission probability based only on its own location, and show that it converges to the optimal proportionally fair policy as the size of the network goes to infinity.


The remainder of the paper is organized as follows. In Section \ref{sec:model} we introduce the system model and AoI more formally. In Section \ref{sec:achievability} we analyze the achievable AoI region. In Section \ref{sec:perfectinfo} we design policies using perfect topology information and derive performance bounds. In Section \ref{sec:unknownconf} we derive a topology-agnostic policy and show asymptotic results. Finally, in Section \ref{sec:simulations} we show simulations supporting our results.

\section{System Model}
\label{sec:model}

We consider a network of $N$ devices, each sending status updates wirelessly to a central base station. Each device, or node, is located at some location in two dimensional space, and we assume that these locations are fixed over the time horizon of interest. Let $r_i$ denote the distance between node $i$ and the base station, where nodes are located inside a circle of a fixed radius. Without loss of generality, we scale distances so that $r_i$ takes values from 0 to 1, and refer to the vector of distances $\boldsymbol{r}$ as the \textit{position vector}.

We assume that update packets have a fixed length, and that time is broken into discrete slots. The duration of a slot is the time required to send one packet. Nodes operate in a random access fashion where node $i$ attempts transmission in each slot according to a Bernoulli process with probability $p_i$, independent of other nodes and across time slots. The process for each node is stationary, so the vector $\boldsymbol{p}$ is fixed across time, but its components can vary between nodes. The network is saturated, meaning nodes sample their process in every time slot and always have an update to send.

As in many wireless systems, a single communication channel is shared by the nodes. When two or more nodes try to access the channel in the same time slot, their signals will interfere and one or more of the transmissions may fail. Because nodes transmit according to a random process, the interference and success of transmissions is also a random process. Let $\tau_i$ be the success probability for node $i$ in a given slot and note that it is stationary under the random access model.

The success probability $\tau_i$ under the capture model is defined as the probability that the signal to interference plus noise ratio is larger than a known threshold $\theta$. Assuming that every node transmits at the same unit power level, the signal strength seen at the base station is a function of the signal attenuation over distance, with roll-off parameter $\beta$, and Rayleigh fading modeled as a random variable $K^2 \sim Exp(1)$. We assume noise is negligible relative to interference, so that the network operates in the interference limited regime, and consider only the signal to interference (SIR) ratio. This does not change any fundamental results on achieveability and provides cleaner analysis throughout.

Under these assumptions, the success probability of node $i$ is the product of its attempt probability and conditional success probability, because each node transmits independently. This can be expressed as
\begin{equation}
\label{eq:throughput1}
    \tau_i = p_i \cdot \mathbb{P}[SIR_i > \theta] = p_i \cdot \mathbb{P} \Bigg[\frac{r_i^{-\beta} K_i^2}{\sum_{j \in I_i} r_j^{-\beta} K_j^2 V_j} > \theta \Bigg],
\end{equation}

where $V_j \sim Ber(p_j)$ is a random variable indicating whether node $j$ transmits in the current time slot, and $I_i$ is the set of nodes which interfere with node $i$, here assumed to be every other node in the network. Following the same procedure as \cite{baccelli2013adaptive}, by conditioning on the Rayleigh fading and Bernoulli transmissions of interferers, and using the complementary CDF of the exponential distribution,
\begin{align}
\begin{aligned}
    \hspace{-1pt} \mathbb{P}[SIR_i > \theta \ | \ V_j, K_j^2, \ \forall j \in I_i] &= e^{-\theta r_i^{\beta} \sum_{j \in I_i} (r_j^{-\beta} K_j^2 V_j)} \\
    &= \prod_{j \in I_i} e^{-\theta r_i^{\beta} r_j^{-\beta} K_j^2 V_j}.
\end{aligned}
\end{align}

Averaging over the Rayleigh fading,
\begin{equation}
    \mathbb{P}[SIR_i > \theta \ | \ V_j, \ \forall j \in I_i] = \prod_{j \in I_i} \frac{1}{1 + \theta r_i^{\beta} V_j r_j^{-\beta}}.
\end{equation}

Now, averaging over the Bernoulli transmissions,
\begin{equation}
    \mathbb{P}[SIR_i > \theta] = \prod_{j \in I_i} \bigg(1 - \frac{p_j}{1 + r_j^{\beta}/r_i^{\beta} \theta} \bigg).
\end{equation}

Finally, by plugging this result back into \eqref{eq:throughput1},
\begin{equation}
\label{successprob}
    \tau_i = p_i \prod_{j \in I_i} \bigg(1 - \frac{p_j}{1 + r_j^{\beta}/r_i^{\beta} \theta} \bigg) = p_i \prod_{j \in I_i} \bigg(1 - \frac{p_j}{1 + d_{ij}} \bigg)
\end{equation}

where $d_{ij} \triangleq r_j^{\beta}/(r_i^{\beta} \theta)$.

We now formally define the Age of Information of node $i$ at time $t$ as $A_i(t)$. This quantity measures the time since the last update received by the base station from node $i$ was generated. Because nodes always generate fresh updates, each packet is generated at the beginning of the time slot in which it is received, and the AoI of node $i$ evolves as
\begin{equation}
    A_i(t+1) = \begin{cases}
        A_i(t) + 1, & \text{if } s_i(t) \neq 1 \\
        1, & \text{if } s_i(t) = 1
    \end{cases}
\end{equation}

where $s_i(t) = 1$ if node $i$ successfully transmits in slot $t$.

\begin{figure}
    \centering
    \includegraphics[width=0.4\textwidth]{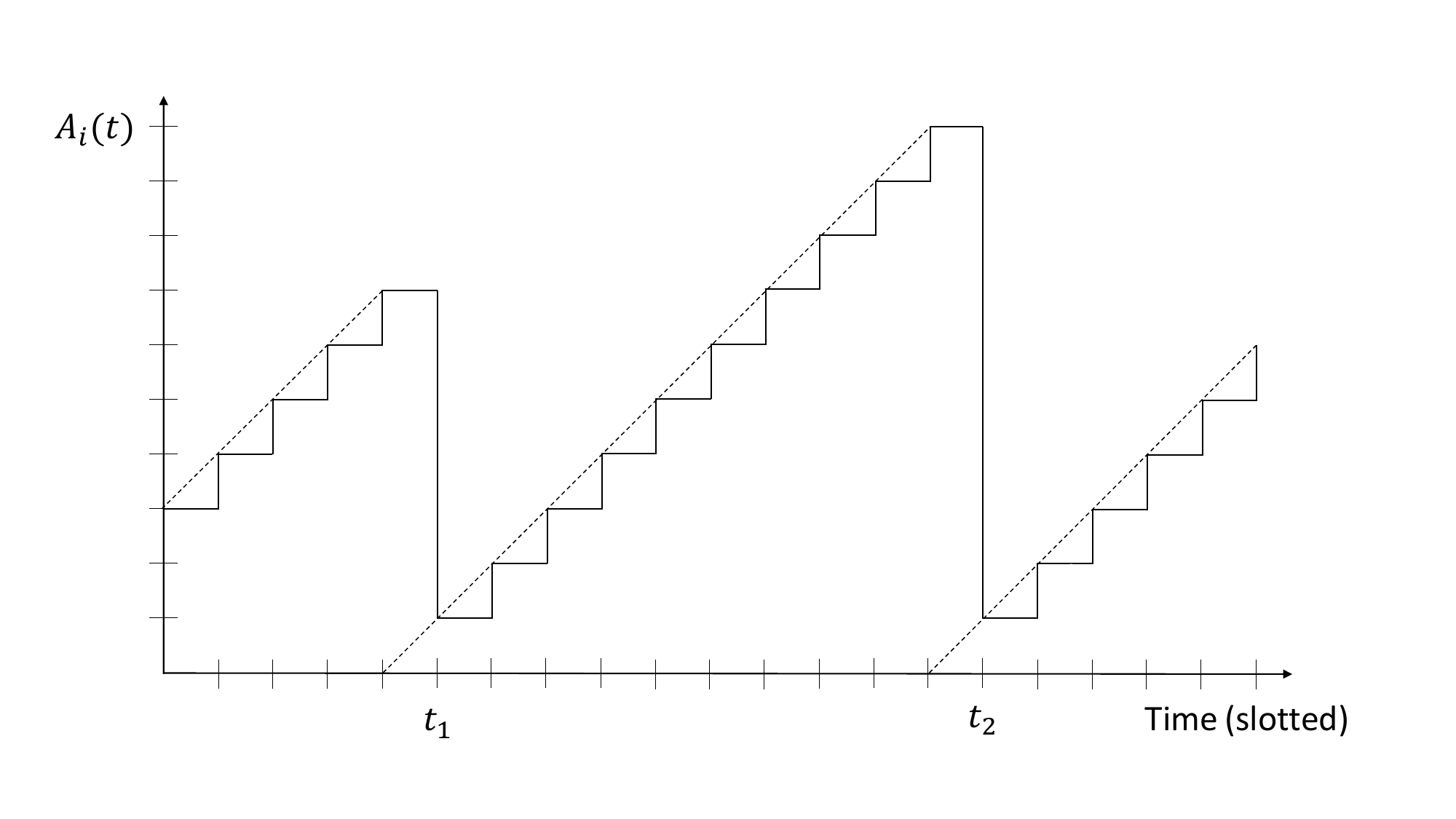}
    \caption{Evolution of the AoI of node $i$ with successes at $t_1$ and $t_2$}
    \label{fig:my_label}
\end{figure}

The infinite time average expected AoI is defined as
\begin{equation}
    h_i \triangleq \lim_{T \to \infty} \frac{1}{T} \sum_{t=1}^T \mathbb{E}[A_i(t)].
\end{equation}

The success probability of each node $i$ is iid Bernoulli with parameter $\tau_i$, so the inter-arrival time $X_i$ between successes is iid Geometric with parameter $\tau_i$. In \cite{kadota2018scheduling}, the authors show that the AoI process is a Renewal process, and from the Renewal Reward theorem \cite[Sec~5.7]{gallager2013stochastic},
\begin{align}
\label{eq:ageofi}
    \begin{aligned}
        h_i &= \frac{\mathbb{E}[X_i^2]}{2 \mathbb{E}[X_i]} + \frac{1}{2} = \frac{(2-\tau_i) \tau_i}{2 \tau_i^2} + \frac{1}{2} = \frac{1}{\tau_i} \\
        &= \frac{1}{p_i \prod_{j \in I_i} \big(1 - \frac{p_j}{1 + d_{ij}} \big)}.
    \end{aligned}
\end{align}

This equation shows that $h_i$ is uniquely determined by the position vector $\boldsymbol{r}$ and the transmission probabilities $\boldsymbol{p}$. We define the vector function $\Phi$ such that
\begin{equation}
    \boldsymbol{h} = \Phi(\boldsymbol{p}, \boldsymbol{r})
\end{equation}

and each element $h_i = \phi_i(\boldsymbol{p}, \boldsymbol{r})$ is given by \eqref{eq:ageofi}.

We are interested in the behavior of $\boldsymbol{h}$, how it is affected by spatial diversity in the network, and how to design policies that minimize AoI in the presence of this diversity. We begin by characterizing the set of achievable AoI in the next section, and will see that understanding this region helps us gain intuition and derive efficient policies and performance bounds.

\section{AoI Achievability Region}
\label{sec:achievability}

We define the set of achievable $\boldsymbol{h}$ for a given position vector $\boldsymbol{r}$ as
\begin{multline}
    H(\boldsymbol{r}) \triangleq \{ \boldsymbol{h} \in \mathbb{R}^N \ | \ h_i = \phi_i(\boldsymbol{p}, \boldsymbol{r}), \ 
    0 \leq p_i \leq 1, \ \forall \ i \} ,
\end{multline}

and the set of all possible $H(\boldsymbol{r})$ as
\begin{equation}
    \mathcal{H} \triangleq \{ H(\boldsymbol{r}) \ | \ 0 < r_i \leq 1, \ \forall \ i \}.
\end{equation}

When discussing general results that hold for any achievable region and are not explicitly dependent on $\boldsymbol{r}$, the index is dropped, and $H$ becomes a generic set in $\mathcal{H}$ with arbitrary position vector.

Clearly, any vector $\boldsymbol{h} \in H$ has components that are positive and always greater than $1$, so $H$ exists in the positive orthant with this lower bound in every dimension. Furthermore, because the AoI of a node can only reach $1$ when it succeeds with probability $1$ in each slot, in order for $H$ to reach $1$ in one dimension it must go to infinity in every other.

Intuitively, under an efficient policy, adjusting $\boldsymbol{p}$ to decrease the AoI of one node will increase the AoI of another. To formalize this notion, we define a set of fixed point equations mapping weights to a policy vector. Let
\begin{equation}
\label{eq:paretoboundary}
	f_i(p_i) \triangleq \frac{\lambda_i}{p_i} - \sum_{j \in I_i} \frac{\lambda_j}{1 + d_{ji} - p_i} = 0, \ \forall \ i
\end{equation}

for $p_i \in [0,1]$ and where $\lambda_i$ is the weight associated with $p_i$. For any set of weights $\boldsymbol{\lambda}$ in the positive orthant, one can use \eqref{eq:paretoboundary} to find the associated $\boldsymbol{p}$ vector. This set of functions plays an important role both in characterizing the boundary of $H$ and designing efficient policies, as shown next.

\begin{lemma}
\label{prop:pareto}
    For any position vector $\boldsymbol{r}$, there exists a Pareto boundary $H^*(\boldsymbol{r})$ to the set $H(\boldsymbol{r})$. For any vector of weights $\boldsymbol{\lambda}$ in the positive orthant, let
    \begin{equation}
        \tilde{p_i} = \min \{ p_i, 1 \}, \ \forall \ i,
    \end{equation}
    
    where $\boldsymbol{p}$ is the unique solution to \eqref{eq:paretoboundary}. The resulting AoI vector $\boldsymbol{h} = \Phi(\boldsymbol{\tilde{p}}, \boldsymbol{r}) \in H^*(\boldsymbol{r})$, i.e. it lies on the Pareto boundary.
    
    Furthermore, for every $\boldsymbol{h}^* = \Phi(\boldsymbol{p}^*, \boldsymbol{r}) \in H^*(\boldsymbol{r})$, there exists a vector $\boldsymbol{\lambda}$ in the positive orthant whose components sum to 1 and for which \eqref{eq:paretoboundary} yields the solution $\boldsymbol{p}^*$ when $p_i^* < 1$ for all $i$.
\end{lemma}

\begin{proof}
    This follows from \cite{gupta2012throughput}, where the authors showed that these results hold for throughput. Because of the inverse relationship between throughput and AoI, every point on the throughput Pareto boundary has a one-to-one mapping to a point in $H$. It is easy to see that these points form a Pareto boundary of $H$, denoted by $H^*$, and that this mapping holds in both directions. Because the results were shown to hold for the throughput boundary, they must also hold for $H^*$.
\end{proof}

    

    
    
    


\begin{corollary}
\label{cor:achievability}
    For a given position vector $\boldsymbol{r}$, an AoI vector $\boldsymbol{h}$ is achievable, i.e. belongs to the set $H(\boldsymbol{r})$, if and only if it lies on or above the Pareto boundary $H^*(\boldsymbol{r})$.
\end{corollary}

\begin{proof}
    We first show the forward direction, that if some $\boldsymbol{h}$ is achievable then it lies on or above $H^*$. Assume it did not. Then it must lie below some $\boldsymbol{h}' \in H^*$ such that $h_i \leq h_i'$ for all $i$, so $\boldsymbol{h}'$ cannot be Pareto optimal, which is a contradiction.
    
    To show the reverse direction, note that every point which lies above the Pareto boundary is in the interior of $H$, because $H$ extends to infinity. Every point on the boundary is achievable from Lemma \ref{prop:pareto}, and just as any point on the interior of the throughput region is achievable \cite{massey1985collision}, so is any point on the interior of $H$. Therefore, every point on or above $H^*$ is achievable.
\end{proof}

We now present the main result of this section.

\begin{theorem}
\label{th:convexH}
    The set $H(\boldsymbol{r})$ is convex for any $N$ and any position vector $\boldsymbol{r}$
\end{theorem}

\begin{proof}
    We start with a geometric proof in two dimensions. Note that we cannot simply prove convexity by randomizing over policies, because random access policies are decentralized, and this randomization technique would require centralized coordination.
    
    In two dimensions and for a fixed $\boldsymbol{r}$, the vector $\boldsymbol{h} = (h_1, h_2)$ is completely characterized by the vector $\boldsymbol{p} = (p_1, p_2)$. In particular, from \eqref{eq:ageofi} both $h_1$ and $h_2$ are convex functions of $\boldsymbol{p}$.
    
    Consider any two points $\boldsymbol{h}^1 = (h_1^1, h_2^1)$ and $\boldsymbol{h}^2 = (h_1^2, h_2^2)$ on $H^*$. Each of these points is uniquely determined by some $\boldsymbol{p}^1$ and $\boldsymbol{p}^2$. Furthermore, because $\phi_i$ is a continuous function of $\boldsymbol{p}$, there is a continuous curve traced from $\boldsymbol{h}^1$ to $\boldsymbol{h}^2$ by
    \begin{equation}
    \label{eq:htrajectory}
        \big( \phi_1(\lambda \boldsymbol{p}^1 + (1-\lambda) \boldsymbol{p}^2, \boldsymbol{r}), \  \phi_2(\lambda \boldsymbol{p}^1 + (1-\lambda) \boldsymbol{p}^2, \boldsymbol{r}) \big),
    \end{equation}
    
    for $0 \leq \lambda \leq 1$. Without loss of generality let $h_1^1 \leq h_1^2$. Because $H^*$ is a Pareto boundary (from Lemma \ref{prop:pareto}), this implies $h_2^1 \geq h_2^2$, and the boundary is monotonically decreasing in the $(h_1, h_2)$ plane. Because $\phi_i$ is convex in $\textbf{p}$, by definition
    \begin{align}
        \begin{aligned}
            \phi_i(\lambda \boldsymbol{p}^1 + (1-\lambda) \boldsymbol{p}^2, \boldsymbol{r}) &\leq \lambda \phi_i(\boldsymbol{p}^1, \boldsymbol{r}) + (1-\lambda) \phi_i(\boldsymbol{p}^2, \boldsymbol{r}) \\
            &= \lambda h_i^1 + (1-\lambda) h_i^2,
        \end{aligned}
    \end{align}
    
    for $i = 1,2$ and $0 \leq \lambda \leq 1$, so each point on the curve is element wise less than or equal to the convex combination of $\boldsymbol{h}^1$ and $\boldsymbol{h}^2$. The points on the curve are by definition achievable, so the achievable region must exist below the tangent connecting $\boldsymbol{h}^1$ and $\boldsymbol{h}^2$, and contain the points on the curve. This holds for any two points $\boldsymbol{h}^1$ and $\boldsymbol{h}^2$, so the set $H$ must be convex in two dimensions. This argument is shown graphically in Figure \ref{fig:h_convexity}. The general proof for $N$ dimensions is given in Appendix \ref{app:convexity}.
    
    \begin{figure}
        \centering
        \includegraphics[width=0.3\textwidth]{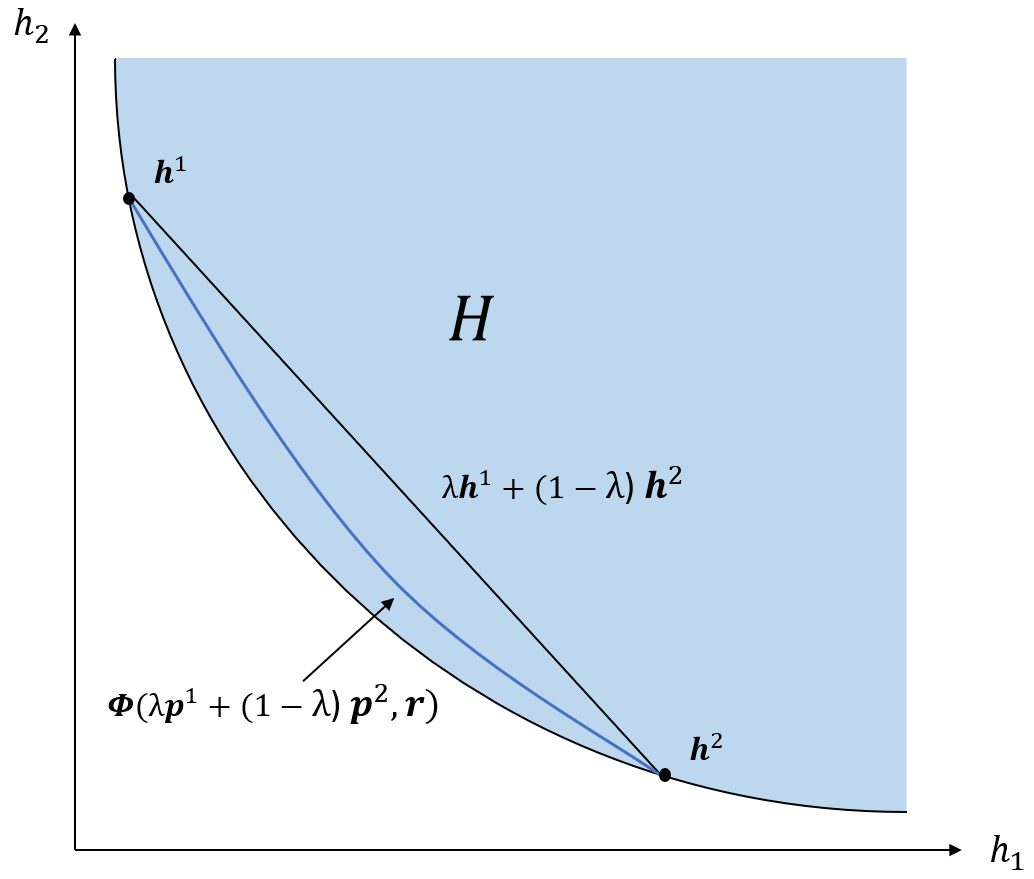}
        \caption{Curve traced between two points in $H^*$ showing convexity of the region $H$}
        \label{fig:h_convexity}
    \end{figure}
\end{proof}


\section{Optimizing AoI with Perfect Topology Information}
\label{sec:perfectinfo}

In the previous section, we characterized the achievable region of time average AoI, and showed that it is convex and has a Pareto boundary. Next we focus on designing policies to operate at specific points in that region. Intuitively, one would expect every policy vector of interest to lie on the Pareto boundary. We will see that this is the case and that one can operate at different points on the boundary to achieve different objectives.

In this section we assume perfect knowledge of the network topology, i.e. the position vector $\boldsymbol{r}$. Because interferer locations play a large role in the interference ratio $1/(1+d_{ij})$, and by extension play a large role in $h_i$, knowledge of $\boldsymbol{r}$ allows us to use spatial diversity to our advantage in designing policies.

\subsection{Expected Weighted Sum AoI}

We begin with the problem of minimizing the expected weighted sum AoI (EWSAoI) of the network over the probability vector $\textbf{p}$,
\begin{align}
    \begin{aligned}
        &\min_{\boldsymbol{p}} \ \sum_{i=1}^N \alpha_i h_i \\
        &\ \text{s.t.} \ 0 \leq p_i \leq 1, \ \forall \ i,
    \end{aligned}
\end{align}

where $\boldsymbol{\alpha} = (\alpha_i, \dots, \alpha_N)$ is the set of positive weights guiding the minimization. This is perhaps the most natural optimization problem for information freshness, and appears commonly in the AoI literature \cite{kadota2018scheduling} \cite{talak2018distributed} \cite{chen2022age}.

It is important to note that while expected AoI is the inverse of expected throughput for a single node, minimizing EWSAoI is not the same as maximizing expected weighted sum throughput. Consider the example of two nodes with $p_1 = 1$ and $p_2 = 0$. The sum throughput of this network is 1, which is optimal for single packet reception, but the expected AoI of node 2 grows without bound, driving the EWSAoI to infinity. This reinforces the motivation to use this metric as opposed to throughput.

Rewriting the problem in terms of \eqref{eq:ageofi}, it becomes
\begin{align}
\label{eq:ewsinp}
    \begin{aligned}
        \min_{\boldsymbol{p}} \ &\sum_{i=1}^N \frac{\alpha_i}{p_i \prod_{j \in I_i} \big(1 - \frac{p_j}{1 + d_{ij}} \big)} \\
        \text{s.t.} \ &0 \leq p_i \leq 1, \ \forall \ i.
    \end{aligned}
\end{align}

Note that this is a convex function in $\textbf{p}$ minimized over a convex set. As a result, this is a convex optimization problem and can be solved using a number of algorithms. Nevertheless, to gain insight into the solution we derive a closed-form expression.

\begin{theorem}
\label{th:ewsaoi}
    The solution to the EWSAoI minimization problem is given by
    \begin{equation}
        p_i^{EWS} = \min \{ \tilde{p_i}, 1 \} , \ \forall \ i,
    \end{equation}
    
    where $\tilde{p_i}$ is the solution to the fixed point equation
    \begin{equation}
    \label{eq:ewsopt}
        \frac{\alpha_i h_i}{\tilde{p_i}} - \sum_{j \in I_i} \frac{\alpha_j h_j}{1 + d_{ji} - \tilde{p_i}} = 0,
    \end{equation}

    and where $h_i$ is the resulting time average expected AoI of node $i$ under this policy.
\end{theorem}

Before proving this result, we first note that the structure of the solution is the same as \eqref{eq:paretoboundary}. Therefore, from Lemma \ref{prop:pareto} it must lie on the Pareto boundary of $H$ for any positive vector of weights, as expected.

\begin{proof}
    As noted previously, \eqref{eq:ewsinp} is a convex optimization problem. Optimization theory tells us that the unconstrained problem has a unique solution, which can be found by setting the gradient equal to zero and solving for $\boldsymbol{p}$ \cite{boyd2004convex}. If the solution satisfies the constraints, then it also solves the constrained problem, otherwise the solution lies on the boundary of the constraint set. In this form, however, the product in the denominator makes it intractable to find a closed-form solution, even for small values of $N$.
    
    By moving each term in the sum into a separate constraint and taking the log, the equivalent problem
    \begin{align}
        \begin{aligned}
            \min_{\boldsymbol{p}} \ &\sum_{i=1}^N \alpha_i h_i' \\
            \text{s.t.} \ &\log h_i' \geq \log \frac{1}{p_i \prod_{j \neq i} \big(1 - \frac{p_j}{1 + d_{ij}} \big)}, \ \forall \ i \\
            &0 \leq p_i \leq 1, \ \forall \ i
        \end{aligned}
    \end{align}
    
    is derived. Note that this is equivalent because minimizing the sum in the objective will drive each constraint on $\boldsymbol{h'}$ to equality. The Lagrangian dual of this problem is given by
    \begin{align}
        \begin{aligned}
            \min_{\boldsymbol{p}} \ &\sum_{i=1}^N \alpha_i h_i' + \sum_{i=1}^N \lambda_i \Bigg(\log \frac{1}{p_i \prod_{j \in I_i} \big(1 - \frac{p_j}{1 + d_{ij}} \big)} - \log h_i' \Bigg) \\
            \text{s.t.} \ &0 \leq p_i \leq 1, \ \lambda_i \geq 0, \ \forall \ i,
        \end{aligned}
    \end{align}
    
    where $\boldsymbol{\lambda}$ is the vector of Lagrange multipliers. This can be simplified further to
    \begin{align}
    \label{eq:ewsaoidual}
        \begin{aligned}
            \min_{\boldsymbol{p}} \ &\sum_{i=1}^N \Bigg( \alpha_i h_i' - \lambda_i \bigg(\log p_i\\
            &\hspace{4em} + \sum_{j \in I_i} \log \bigg(1 - \frac{p_j}{1 + d_{ij}} \bigg) + \log h_i' \bigg) \Bigg) \\
            \text{s.t.} \ &0 \leq p_i \leq 1, \ \lambda_i \geq 0, \ \forall \ i.
        \end{aligned}
    \end{align}
    
    Maximizing the solution to this problem over $\boldsymbol{\lambda}$ yields the solution to the dual problem, which is a lower bound on the solution to the primal. Because the primal problem is convex and the feasible set has a non-empty interior, Slater's condition is satisfied. Therefore, strong duality holds and this bound is tight \cite[Sec~5.2]{boyd2004convex}, making the maximum of this problem over $\boldsymbol{\lambda}$ equivalent to the primal problem.
    
    Because \eqref{eq:ewsaoidual} is convex in $\boldsymbol{p}$ and $\boldsymbol{h'}$, and concave in $\boldsymbol{\lambda}$, the solution is found by taking the gradient with respect to each and setting it equal to zero, thereby maximizing over $\boldsymbol{\lambda}$ and minimizing over $\boldsymbol{p}$ and $\boldsymbol{h'}$. This gives the solution
    \begin{equation}
        \begin{cases}
        \begin{aligned}
            &\frac{\lambda_i}{p_i} = \sum_{j \in I_i} \dfrac{\lambda_j}{1 + d_{ji} - p_i} \\
            &\lambda_i = \alpha_i h_i' \\
            &\log h_i' = -\log p_i - \sum_{j \in I_i} \log \bigg(1 - \dfrac{p_j}{1 + d_{ij}} \bigg)
        \end{aligned}
        \end{cases}
    \end{equation}
    
    for all $i$ and when $p_i \in [0,1]$. By taking the log of \eqref{eq:ageofi}, it can immediately be seen that $h_i' = h_i$ under the resulting policy, and that the Lagrange multipliers are equal to the weighted expected AoI. Combining these three equations yields \eqref{eq:ewsopt}.
    
    To see that this solution is unique, note that the left hand side of \eqref{eq:ewsopt} is monotonically decreasing in $p_i$ and goes to infinity as $p_i$ goes to $0$. If it becomes negative when $p_i = 1$, then by the Intermediate Value Theorem there exists a unique solution in the domain of $p_i$ given by \eqref{eq:ewsopt}. Otherwise the minimum occurs at $p_i = 1$ by the monotonicity of \eqref{eq:ewsopt}.
\end{proof}

\subsection{Min-Max Expected AoI}

Minimizing EWSAoI achieves some notion of fairness in the network, because the optimization will not let any single node's AoI grow too large, but we may sometimes be interested in a clearer notion of fairness. Min-max fairness for AoI is equivalent to max-min fairness for throughput and ensures the most evenly minimized AoI across the network. In particular, min-max optimization is defined as minimizing the maximum AoI in the network over the vector $\boldsymbol{p}$,
\begin{align}
\label{eq:mminp}
    \begin{aligned}
        \min_{\textbf{p}} \ &\max \ \{ h_i \} \\
        \text{s.t.} \ &\ 0 \leq p_i \leq 1, \ \forall \ i.
    \end{aligned}
\end{align}

This solution follows along similar lines to Theorem \ref{th:ewsaoi}.




\begin{theorem}
\label{th:minmaxsolution}
    The solution to the expected min-max AoI (MMAoI) problem is given by
    \begin{equation}
        p_i^{MM} = \min \{ \tilde{p_i}, 1 \} , \ \forall \ i,
    \end{equation}
    
    where $\tilde{p_i}$ is the solution to the fixed point equation
    \begin{equation}
    \label{eq:mmopt}
        \frac{\lambda_i}{\tilde{p_i}} - \sum_{j \in I_i} \frac{\lambda_j}{1 + d_{ji} - \tilde{p_i}} = 0,
    \end{equation}

    and $\boldsymbol{\lambda}$ is such that its entries sum to the log of the resulting AoI, $\log h^{MM} = \log h_i^{MM}$, for all $i$.
\end{theorem}

\begin{proof}
    See Appendix \ref{app:mmaoi}
\end{proof}

Note that the solution takes the form of \eqref{eq:paretoboundary}, and so lies on the Pareto boundary as expected. We also note that the resulting AoI is equal across all nodes, and is denoted by $h^{MM}$. The performance is shown through simulations in Section \ref{sec:simulations}.

\subsection{Proportionally Fair Expected AoI}

The last optimization metric of interest is proportionally fair AoI (PFAoI). Proportional fairness is defined as the maximum sum log of throughput \cite{kelly1997charging}, so equivalently is the minimum sum log of AoI,
\begin{align}
    \begin{aligned}
        &\min_{\boldsymbol{p}} \ \sum_{i=1}^N \log h_i \\
        &\ \text{s.t.} \ 0 \leq p_i \leq 1, \ \forall \ i.
    \end{aligned}
\end{align}

This can be rewritten in terms of \eqref{eq:ageofi} as
\begin{align}
\label{eq:pfinp}
    \begin{aligned}
        \min_{\boldsymbol{p}} \ &\sum_{i=1}^N \log \frac{1}{p_i \prod_{j \in I_i} \big(1 - \frac{p_j}{1 + d_{ij}} \big)} \\
        \text{s.t.} \ &0 \leq p_i \leq 1, \ \forall \ i.
    \end{aligned}
\end{align}

In similar fashion to EWSAoI and MMAoI, we derive a closed-form solution to this problem, and in fact see that it takes a simpler form.

\begin{theorem}
\label{th:pfsol}
    The solution to the PFAoI minimization problem is given by
    \begin{equation}
    \label{eq:pfopt}
        p_i^{PF} = \min \{ \tilde{p_i}, 1 \} , \ \forall \ i,
    \end{equation}
    
    where $\tilde{p_i}$ is the solution to the fixed point equation
    \begin{equation}
        \frac{1}{\tilde{p_i}} - \sum_{j \in I_i} \frac{1}{1 + d_{ji} - \tilde{p_i}} = 0, \ \forall \ i.
    \end{equation}
\end{theorem}

Once again we note that the solution takes the form of \eqref{eq:paretoboundary} with weights all equal to $1$, and so lies on the Pareto boundary from Lemma \ref{prop:pareto}.

\begin{proof}
    The problem in \eqref{eq:pfinp} can be rewritten as
    \begin{align}
        \begin{aligned}
            \min_{\boldsymbol{p}} \ &\sum_{i=1}^N \Bigg( - \log p_i - \sum_{j \in I_i} \log \bigg(1 - \frac{p_j}{1 + d_{ij}} \bigg) \Bigg) \\
            \text{s.t.} \ &0 \leq p_i \leq 1, \ \forall \ i,
        \end{aligned}
    \end{align}

    which is convex due to the convexity of the equation in $\boldsymbol{p}$ and the convexity of the set. Therefore, the minimum can be found directly by taking the gradient and setting it equal to zero, immediately yielding the result \eqref{eq:pfopt}. Uniqueness follows along similar lines to Theorem \ref{th:ewsaoi}.
    
\end{proof}

Not only does PFAoI have a closed-form solution similar to EWSAoI and MMAoI, it also has the advantage of being completely separable. This means the optimization can be performed in a distributed manner, with each node computing its own transmission probability and without sharing Lagrange multiplier values between nodes.


We conjecture that PFAoI will have similar performance to EWSAoI with symmetric weights. From Theorem \ref{th:ewsaoi}, the Lagrange multipliers in the EWSAoI solution with symmetric weights are equal to the resulting AoI of each node, and we know heuristically that when minimizing the sum, no single node's AoI will grow too large. As a result, the multipliers will be close to equal. Because scaling them all by a constant has no effect on the solution to \eqref{eq:paretoboundary}, this is the same as if they were all close to 1, i.e. the PFAoI solution \eqref{eq:pfopt}.

We will see through simulations that this is indeed the case, and while we provide no rigorous guarantees on performance, this heuristic argument combined with the separability of PFAoI makes it an appealing alternative to EWSAoI. We can, however, provide performance guarantees on EWSAoI and MMAoI, which we show next.

\subsection{Performance Bounds}

We denote the solution to the EWSAoI problem with symmetric weights by $\boldsymbol{p}^S(\boldsymbol{r})$, and the corresponding AoI vector by $\boldsymbol{h}^S(\boldsymbol{r})$. Here the dependence on the position vector $\boldsymbol{r}$ is shown explicitly.

\begin{theorem}
\label{th:perfinfobound}
    For any position vector $\boldsymbol{r}$, and with $\beta=2$ and $\theta=1$, the normalized AoI is bounded such that
    \begin{equation}
        1 \leq \frac{1}{N^2} \sum_{i=1}^N h_i^S(\boldsymbol{r}) \leq \frac{1}{N} h^{MM}(\boldsymbol{r}) \leq \frac{e}{2}
    \end{equation}
    
    as the size of the network $N$ goes to infinity. This upper bound is tight when all nodes are fixed on a circle an equal distance away from the base station.
\end{theorem}

\begin{proof}
    See Appendix \ref{app:perfinfobound}

    
\end{proof}

This bound provides tight guarantees on average AoI, and in the case of $h^{MM}$, fairness guarantees that \textit{every} node will lie below this bound. The driving motivation behind this work is the lack of fairness and the poor performance achieved by traditional random access policies in the presence of spatial diversity, but this result shows that when spatial diversity is built into the policy, it actually improves performance. Simulations in Section \ref{sec:simulations} further verify these results. In the next section, we examine the case where nodes do not have access to the position vector $\boldsymbol{r}$.

\section{Optimizing Topology Agnostic AoI}
\label{sec:unknownconf}

In the previous section, we examined policies that minimize different metrics of AoI under the assumption of perfect topology information, where the entire position vector $\boldsymbol{r}$ is known. In practical settings, this may not be the case. It is realistic to assume that each node has knowledge of its own position through GPS or another technology, but may be blind to the locations of other nodes. Moreover, in a distributed setting, nodes must be able to compute their transmission probabilities independently.

Clearly, without location information one expects some loss in performance. Therefore, our goal is to design a policy which minimizes this loss, while setting transmission probabilities for each node based only on its own location. We refer to this class of policies as \textit{topology agnostic}, and focus on proportionally fair AoI, because the policy \eqref{eq:pfopt} is completely decoupled and hence amenable to distributed implementation. Moreover, as we conjectured and show through simulation results in Section \ref{sec:simulations}, the PFAoI policy serves as a good proxy for minimizing EWSAoI.

\subsection{Topology Agnostic Proportionally Fair AoI}

To model unknown interferer locations, assume that nodes are uniformly distributed in $\mathbb{R}^2$, and that the locations of interfering nodes are random variables. We use capital $R_j$ to denote the distance of node $j$ from the base station. Similar analysis can be done for a non-uniform distribution of nodes, but we restrict our attention to the uniform setting for simplicity.

Because nodes are uniformly distributed over a circle of radius $1$, this circle has an area equal to $\pi$, and the probability that any node lies within some region in space is the area of that region divided by $\pi$. Therefore, the probability that node $j$ lies within a circle of radius $r_j$, equivalent to saying its distance from the base station is less than $r_j$, is
\begin{equation}
    \mathbb{P}[R_j \leq r_j] = r_j^2, \ \forall \ j.
\end{equation}

This is the CDF for the random variable $R_j$, and its pdf is given by its derivative
\begin{equation}
    f_{R_j}(r_j) = 2r_j, \ \forall \ j.
\end{equation}

We define the expected difference between our topology agnostic (TA) and proportionally fair objectives as the optimality gap, which quantifies the loss incurred by not knowing $\boldsymbol{r}$. The optimal TA policy is then given by the solution to
\begin{equation}
\label{eq:pfloss}
    \min_{\pi \in \Pi} \ \sum_{i = 1}^N \mathbb{E}[ \log h_i^{\pi}] - \sum_{i=1}^N \log h_i^{PF},
\end{equation}

where $\Pi$ is the class of topology agnostic policies. Here the expectation in each term of the sum is taken with respect to the interferer locations for that term, i.e. the expectation of $h_i$ is taken with respect to all $R_j$ for $j \in I_i$. The resulting policy is simple and elegant, as shown in the following theorem.

\begin{theorem}
    For large $N$, when nodes are uniformly distributed in $\mathbb{R}^2$, and when $\beta=2$ and $\theta=1$, the policy TA that minimizes the optimality gap \eqref{eq:pfloss} to proportionally fair AoI is given by
    \begin{equation}
    \label{eq:pipi}
        p_i^{TA} \approx \frac{1}{(N-1) \big(1 - r_i^2 \log (1 + \frac{1}{r_i^2}) \big)}, \ \forall \ i.
    \end{equation}
\end{theorem}

\begin{proof}
    The problem \eqref{eq:pfloss} is equivalent to minimizing the first summation, so the second sum can be dropped, and the problem can be rewritten in terms of \eqref{eq:ageofi} as
    \begin{align}
        \begin{aligned}
            \min_{\boldsymbol{p}} \ &\sum_{i = 1}^N \mathbb{E}_{ \boldsymbol{R}} \Bigg[ \log \frac{1}{p_i \prod_{j \in I_i} (1 - \frac{p_j}{1 + D_{ij}})} \Bigg] \\
            \text{s.t.} \ &0 \leq p_i \leq 1, \ \forall \ i,
        \end{aligned}
    \end{align}
    
    where $D_{ij} = R_j^{\beta}/(R_i^{\beta} \theta) = (R_j/R_i)^2$ under our assumptions. Rearranging the objective function further,
    \begin{multline}
            \sum_{i = 1}^N \mathbb{E}_{\boldsymbol{R}} \bigg[ -\log p_i - \sum_{j \in I_i} \log \bigg(1 - \frac{p_j}{1 + (R_j/R_i)^2} \bigg) \bigg] \\
            = \sum_{i = 1}^N \mathbb{E}_{\boldsymbol{R}} \bigg[ -\log p_i - \sum_{j \in I_i} \log \bigg(1 - \frac{p_i}{1 + (R_i/R_j)^2} \bigg) \bigg],
    \end{multline}
    
    where in the re-indexed sums on the right hand side, $p_i$ only appears in one term of the outer sum, for each $i$.
    
    Recall that topology agnostic policies assume a distributed implementation, where each node has access to its own $r_i$. Thus, re-indexing the sum effectively decouples the optimization so that each node can set its $p_i$ by minimizing a single term in the outer sum, independent of other nodes. Because each node has access to its location, the expectation can be conditioned on $r_i$ in each term, and the problem becomes
    \begin{align}
    \label{eq:expectedpf}
        \begin{aligned}
            \min_{\boldsymbol{p}} \ &\sum_{i = 1}^N \mathbb{E}_{\boldsymbol{R}} \bigg[ -\log p_i - \sum_{j \in I_i} \log \bigg(1 - \frac{p_i}{1 + \frac{R_i^2}{R_j^2}} \bigg) \ | \ R_i = r_i \bigg] \\
            \text{s.t.} \ &0 \leq p_i \leq 1, \ \forall \ i.
        \end{aligned}
    \end{align}
    
    A single term of the sum with its expectation written explicitly is equal to
    \begin{equation}
        -\log p_i - \sum_{j \in I_i} \ \int_0^1 2r_j \log \bigg(1 - \frac{p_i}{1 + r_i^2/r_j^2} \bigg) dr_j.
    \end{equation}
    
    This integral does not have a closed-form solution, but can be closely approximated using $\log(1-x) \approx -x$ for small $x$ to eliminate the log in the integral. This is valid for large $N$, because $p_i$ will be small and the denominator is greater than $1$. Then
    \begin{multline}
        \sum_{j \in I_i} \int_0^1 2r_j \log \bigg(1 - \frac{p_i}{1 + r_i^2/r_j^2} \bigg) dr_j \\
        \begin{aligned}
            &\approx - \sum_{j \in I_i} \int_0^1 2r_j \bigg(\frac{p_i}{1 + r_i^2/r_j^2} \bigg) dr_j \\
            &= - (N-1) p_i \big(1 - r_i^2 \log (1 + 1/r_i^2 ) \big).
        \end{aligned}
    \end{multline}
    
    Plugging this result back into \eqref{eq:expectedpf}, the optimization becomes
    \begin{align}
        \begin{aligned}
            \min_{\boldsymbol{p}} \ &\sum_{i = 1}^N \big( -\log p_i + (N-1) p_i \big(1 - r_i^2 \log (1 + 1/r_i^2 ) \big) \big) \\
            \text{s.t.} \ &0 \leq p_i \leq 1, \ \forall \ i.
        \end{aligned}
    \end{align}
    
    This problem is convex in $\boldsymbol{p}$, and the solution \eqref{eq:pipi} is found by setting the gradient equal to $0$ and solving. 
\end{proof}

Under this policy, each node sets its transmission probability based only on its own location and the number of nodes in the network, using a clean, closed-form expression. This makes it a simple and attractive policy.

\subsection{Asymptotic Results and Convergence}

Next, we quantify how far the true optimal $p_i^{PF}$ deviates from $p_i^{TA}$ in \eqref{eq:pipi} as $N$ goes to infinity, by computing a probability distribution on $p_i^{PF}$.

\begin{lemma}
\label{th:pfconvergence}
    When nodes are uniformly distributed in $\mathbb{R}^2$, and when $\beta=2$ and $\theta=1$, the inverse of the proportionally fair optimal $p_i^{PF}$ converges to a truncated normally distributed random variable as $N$ goes to infinity. Specifically,
    \begin{equation}
        \label{eq:pfconvergence}
        Z_i^{PF} = \frac{1}{p_i^{PF}} \to \mathcal{N} (N \mu_i, N \sigma_i^2 )
    \end{equation}
    
    on the interval $[1, \infty)$, with a probability mass spike at $1$, and where
    \vspace{-6pt}\begin{align}
        \label{eq:pfmean}
        &\mu_i = 1 - r_i^2 \log \bigg(1 + \frac{1}{r_i^2} \bigg), \\ 
        \label{eq:pfvar}
        &\sigma_i^2 = 1 - \frac{1}{1+r_i^2} - \Bigg( r_i^2 \log \bigg(1 + \frac{1}{r_i^2} \bigg) \Bigg)^2.
    \end{align}
\end{lemma}

\begin{proof}
    Recall from Theorem \ref{th:pfsol} that the optimal proportionally fair expected AoI policy is given by the fixed point equation 
    \begin{equation*}
        Z_i^{PF} = \sum_{j \in I_i} \frac{1}{1 + D_{ji} - p_i} \approx \sum_{j \in I_i} \frac{1}{1 + D_{ji}}, \ \forall \ i
    \end{equation*}
     
    when the resulting $p_i \leq 1$, and where $D_{ji}$ is now a random variable. The approximation is valid because $p_i$ is small when $N$ is large. Conditioned on the value of $r_i$, each term in this sum becomes an iid random variable. As $N$ goes to infinity, applying the Central Limit Theorem shows that the sum converges to a normal distribution,
    \begin{equation}
        Z_i^{PF} \to \mathcal{N} (N \mu_i, N \sigma_i^2 ), \ \forall \ i,
    \end{equation}
    
    where $\mu_i$ and $\sigma_i^2$ are the mean and variance of the iid terms in the sum, conditioned on $r_i$. These can be computed as \eqref{eq:pfmean} and \eqref{eq:pfvar} respectively.
    
    This distribution has a non-zero probability that $Z_i^{PF}$ takes a value larger than $1$, which corresponds to when $\tilde{p}_i > 1$ in \eqref{eq:pfopt}. This clearly cannot happen because $p_i$ is a probability, and as a result the normal distribution in \eqref{eq:pfconvergence} is only valid on the interval $[1, \infty)$. The remaining probability mass is contained in the spike
    \begin{equation}
    \label{eq:pmfmass}
        \mathbb{P} [Z_i^{PF} = 1] = \mathbb{P} \bigg[ R_i^2 \geq \sum_{j \in I_i} R_j^2 \ | \ R_i= r_i \bigg],
    \end{equation}
    
    because these two events are equivalent. This can be shown by evaluating $f_i(1)$ in \eqref{eq:paretoboundary} and noting that $f_i$ is a decreasing function of $p_i$. Then the condition $p_i > 1$ is equivalent to $f(p_i) < f(1)$, and after some algebraic manipulation, this becomes the event in \eqref{eq:pmfmass}. Because $N$ is large, this condition has negligible probability. Nevertheless, it yields a valid probability distribution, with the negligible mass that would have existed in the negative tail of the distribution occurring at this spike.
    
    This subtlety is clearly necessary because probabilities cannot take values larger than $1$. Furthermore, this ensures that the mean and variance of the distribution \eqref{eq:pfconvergence} are well defined as $N$ goes to infinity, because the distribution only takes positive values. This completes the proof.
\end{proof}

Note that the mean of $Z_i^{PF}$ is approximately equal to $1/p_i^{TA}$, so the policy $TA$ which minimizes the expected optimality gap is equivalent to taking the mean of the distribution of $Z_i^{PF}$. As shown next, when $N$ goes to infinity, the optimal $p_i^{PF}$ converges to this mean with high probability.

\begin{theorem}
\label{th:pfconvrate}
    As $N$ becomes large, the random variable $p_i^{PF}$ converges to $p_i^{TA}$ such that
    \begin{equation}
        \mathbb{P} \bigg[ | p_i^{PF} - p_i^{TA} | \leq \frac{k}{N^{3/2}}  \bigg] = 1 - \epsilon
    \end{equation}
    
    for some constant $k$ and any $\epsilon > 0$. In other words, the TA policy $p_i^{TA}$ is asymptotically optimal with arbitrarily high probability, and converges at a rate of $1/N^{3/2}$.
\end{theorem}

\begin{proof}
    
    
    
    
    Because $ Z_i^{PF}$ is normally distributed according to \eqref{eq:pfconvergence}, the event that it lies within $m$ standard deviations of the mean is
    \begin{equation}
    \label{eq:stddevevent}
        |  Z_i^{PF} - N \mu_i | \leq m \sqrt{N} \sigma_i, 
    \end{equation}
    
    and for any fixed value of $m$, the probability that this occurs can be written as
    \begin{equation}
        \mathbb{P} \big[ |  Z_i^{PF} - N \mu_i | \leq m \sqrt{N} \sigma_i \big] = 1 - \epsilon_m,
    \end{equation}
    
    where $\epsilon_m$ decreases with increasing $m$. Substituting $p_i^{PF}$ and solving for it in \eqref{eq:stddevevent},
    \begin{equation}
        \frac{1}{N \mu_i + m \sqrt{N} \sigma_i} \leq p_i^{PF} \leq \frac{1}{N \mu_i - m \sqrt{N} \sigma_i}.
    \end{equation}
    
    Now rearranging the lower bound to be in the form
    \begin{align}
        \begin{aligned}
            \frac{1}{N \mu_i + m \sqrt{N} \sigma_i} &= \frac{1}{N \mu_i} - p_i^L
            \approx p_i^{TA} - p_i^L,
        \end{aligned}
    \end{align}
    
    and doing some algebraic manipulation,
    \begin{align}
        \begin{aligned}
            p_i^L &= \frac{1}{N \mu_i} - \frac{1}{N \mu_i + m \sqrt{N} \sigma_i}
            \leq  \frac{k}{N^{3/2}}
        \end{aligned}
    \end{align}
    
    for any finite $m$ and some constant $k$ as $N$ goes to infinity. Following the same approach for the upper bound $p_i^{TA} + p_i^U$, one can see that $p_i^U$ is on the same order. Therefore, the event
    \begin{equation}
        | p_i^{PF} - p_i^{TA} | \leq \frac{k}{N^{3/2}}
    \end{equation}
    
    is equivalent to \eqref{eq:stddevevent} for any finite $m$, and occurs with probability 
    \begin{equation}
        \mathbb{P} \bigg[ | p_i^{PF} - p_i^{TA} | \leq \frac{k}{N^{3/2}}  \bigg] = 1 - \epsilon_m
    \end{equation}
    
    where $\epsilon_m$ becomes arbitrarily small as $m$ increases. This completes the proof.
\end{proof}

We have shown that with high probability, the optimal proportionally fair policy converges to the policy TA at a rate of $1/N^{3/2}$. It is important to note that this is faster than $p_i^{TA}$ converges to 0, which occurs at a rate of $1/N$. 

Therefore, for large $N$, we expect the policy TA to achieve similar performance to the PFAoI policy (which uses all node locations), and we show through simulations in the next section that this is the case. This is a powerful result, which shows that a decoupled policy in which each node sets its transmission probability based \textit{only on its own location} asymptotically achieves proportional fairness with arbitrarily high probability, and furthermore converges quickly. In the next section, we verify this convergence and the performance of all our policies through simulations.

\section{Numerical Results}
\label{sec:simulations}

In this section, we evaluate the performance of our policies using normalized expected AoI, defined as $h_i/N$, as a comparison metric. In all simulations of EWSAoI, the weights are set equal to $1$. In Figure \ref{fig:policy_comp_vs_N}, the network average normalized AoI is shown as a function of $N$, averaged over 100 random, uniformly distributed topologies. For each sample, the topology was generated and nodes were added incrementally to observe the change with $N$. The EWS, PF, and MM policies all achieve ostensibly the same network average AoI, which approaches a constant value below the EWS and MM upper bound (UB) of $e/2$. As expected, the topology agnostic policy (TA) achieves worse performance for small $N$, but converges quickly to the other policies.

In Figure \ref{fig:p_comp_vs_r}, the transmission probabilities of each policy are shown as a function of distance from the base station. Here one node is fixed at $r_i$ and the rest of the topology is randomly generated from a uniform distribution. The EWS, PF, and MM policies are averaged over 1000 such topologies. We see that the four curves have the same general shape, increasing with $r_i$ as expected, but with small variations. As predicted, $\boldsymbol{p}^{PF}$ and $\boldsymbol{p}^{TA}$ are very close even at $N=50$, demonstrating the convergence in Theorem \ref{th:pfconvrate}.

Finally, Figure \ref{fig:h_comp_vs_r} shows the normalized expected AoI as a function of distance, averaged over the same 1000 topologies, and compared to traditional slotted ALOHA. Although the network average for the four policies is very close, there is some disparity across the network and one can observe how the max AoI changes under the different policies. The MM policy is flat as expected, achieving complete fairness, but the other three policies exhibit a small amount of unfairness. Intuitively, nodes close to the base station perform well because of their high signal strength and nodes far from the base station perform well because of their large transmission probability. This creates the concavity in the plot and causes nodes in the middle to perform the worst. The plot highlights the relative fairness of all four policies over slotted ALOHA, which sees more than a factor of four increase in AoI as $r_i$ increases from $0.125$ to $1$.

As conjectured, the PF policy closely mimics the EWS policy and justifies its use as a proxy. Furthermore, even for small values of $N$, the topology agnostic policy TA achieves similar performance. We conclude that PF is a good proxy for minimizing sum or min-max AoI, while also easily implementable in a distributed system. Furthermore, when topology information is unknown, TA is a good approximation to PF for moderate and larger sized networks.

\begin{figure}
    \centering
    \includegraphics[width=0.42\textwidth]{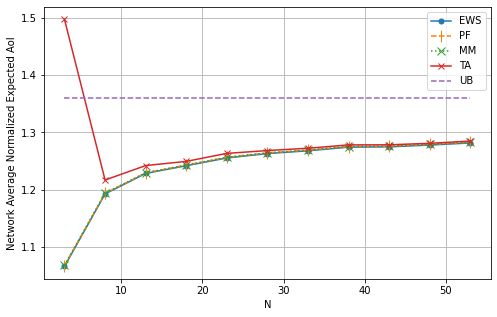}
    \caption{Normalized network average expected AoI for varying $N$ under each policy, averaged over 100 random network topologies}
    \label{fig:policy_comp_vs_N}
\end{figure}

\begin{figure}
    \centering
    \includegraphics[width=0.42\textwidth]{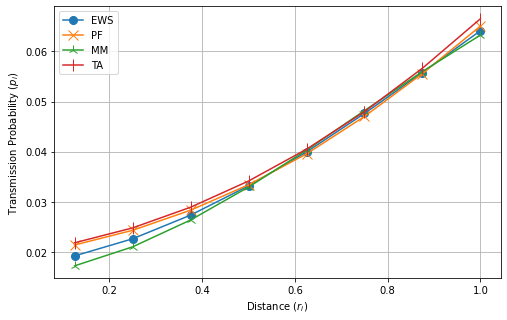}
    \caption{Transmission probability under each policy as a function of $r_i$, for $N=50$ and averaged over 1000 random network topologies}
    \label{fig:p_comp_vs_r}
\end{figure}

\begin{figure}
    \centering
    \includegraphics[width=0.42\textwidth]{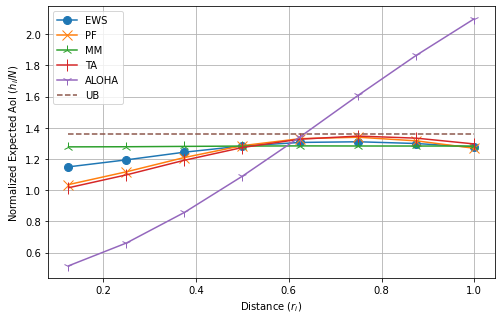}
    \caption{Normalized expected AoI under each policy as a function of $r_i$, for $N=50$ and averaged over 1000 random network topologies}
    \label{fig:h_comp_vs_r}
\end{figure}

\section{Conclusion}

In this paper we examined AoI in spatially distributed networks under a random access policy. We characterized the achievable region of time average AoI and showed that it's convex. We designed policies to minimize weighted sum, min-max, and proportionally fair AoI with full knowledge of the network topology, and showed tight performance bounds on weighted sum and min-max. We further derived a policy to minimize proportionally fair AoI when nodes only have access to their own locations, and showed that as the size of the network goes to infinity, it converges to the proportionally fair policy. We then verified our results through simulations.

Possible future directions include combining this work with threshold ALOHA policies derived in \cite{chen2022age, yavascan2021analysis} or a version of CSMA, extending our analysis to unsaturated networks where nodes don't always have updates to send, and implementing our policies in a real system to compare against 802.11 and other protocols.

\appendix 
\subsection{Proof of Theorem \ref{th:convexH}}
\label{app:convexity}

We have shown that the set $H$ is convex in two dimensions. We now show that this property extends to $N$ dimensions using the convexity of the vector function $\Phi$ and the relation of $H$ to the epigraph of this function. The convexity of a vector valued function is a natural generalization of the convexity of a scalar function \cite{dattorro2010convex}, and holds if and only if the following conditions are met.

$(i)$ The relationship
\begin{equation}
    \Phi(\lambda \boldsymbol{p}^1 + (1-\lambda) \boldsymbol{p}^2, \boldsymbol{r}) \preceq \lambda \Phi(\boldsymbol{p}^1, \boldsymbol{r}) + (1-\lambda) \Phi(\boldsymbol{p}^2, \boldsymbol{r})
\end{equation}
    
is true for all $\boldsymbol{p}^1, \boldsymbol{p}^2$ in the domain of $\Phi$ and all $0 \leq \lambda \leq 1$. Here the symbol $\preceq$ is defined as element-wise comparison.

$(ii)$ The domain of $\Phi$ is a convex set.

In the case of the set $H$, $(i)$ holds by the convexity of $\phi_i$, and $(ii)$ holds because the domain of $\Phi$ is a cube in $\mathbb{R}^N$, where each element takes values from 0 to 1.
    
The epigraph of the vector function $\Phi$ is defined in \cite{dattorro2010convex} as
\begin{equation}
    (\text{epi} \ \Phi)(\boldsymbol{r}) \triangleq \{ (\boldsymbol{p}, \boldsymbol{h}) \in \mathbb{R}^N \times \mathbb{R}^N \ | \ \boldsymbol{p} \in \text{dom} \ \Phi, \ \Phi(\boldsymbol{p}, \boldsymbol{r}) \preceq \boldsymbol{h} \},
\end{equation}
    
or the set of all $(\boldsymbol{p}, \boldsymbol{h})$ pairs such that $\boldsymbol{p}$ is in the domain of $\Phi$ and $\boldsymbol{h}$ is element-wise greater than or equal to $\Phi(\boldsymbol{p}, \boldsymbol{r})$.

$\Phi$ is a convex function, so $\text{epi} \ \Phi$ is a convex set from \cite{dattorro2010convex}. Furthermore, the projection of a convex set onto some of its coordinates is convex \cite{boyd2004convex}, so the projection of $\text{epi} \ \Phi$ onto $H(\boldsymbol{r})$ is convex.

It remains to show that the projection of epi $\Phi$ onto $H(\boldsymbol{r})$ is equal to the set $H(\boldsymbol{r})$. Consider any achievable point $\boldsymbol{h}' \in H(\boldsymbol{r})$. By definition, there exists some $\boldsymbol{p}'$ such that $\boldsymbol{h}' = \Phi(\boldsymbol{p}', \boldsymbol{r})$. Therefore $(\boldsymbol{p}', \boldsymbol{h}') \in \text{epi} \ \Phi$ and $\boldsymbol{h}'$ is in its projection.

To show the opposite direction, consider any $\boldsymbol{h}'$ in the projection of epi $\Phi$. There exists by definition some $\boldsymbol{\tilde{p}} \in \text{dom} \ \Phi$ such that $\boldsymbol{\tilde{h}} = \Phi(\boldsymbol{\tilde{p}}, \boldsymbol{r}) \preceq \boldsymbol{h}'$, and $\boldsymbol{\tilde{h}} \in H(\boldsymbol{r})$. From Corollary \ref{cor:achievability}, a point is achievable iff it lies above some point $\boldsymbol{h}^*$ on the Pareto Boundary, so $\boldsymbol{h}^* \preceq \boldsymbol{\tilde{h}} \preceq \boldsymbol{h}'$, and $\boldsymbol{h}' \in H(\boldsymbol{r})$. This completes the proof.
    
\subsection{Proof of Theorem \ref{th:minmaxsolution}}
\label{app:mmaoi}

Rewriting the problem \eqref{eq:mminp} in terms of \eqref{eq:ageofi} and moving each term in the maximum to the constraint set,
    
\begin{align}
    \begin{aligned}
        \min_{\boldsymbol{p}} \ &\alpha \\
        \text{s.t.} \ &\alpha \geq \frac{1}{p_i \prod_{j \in I_i} (1 - \frac{p_j}{1 + d_{ij}})}, \ \forall \ i \\
        &0 \leq p_i \leq 1, \ \forall \ i.
    \end{aligned}
\end{align}

This problem is convex in $\boldsymbol{p}$, because the constraints form a convex set. As shown in \cite{baccelli2013adaptive}, taking the log of the constraints and setting $\tilde{\alpha} = -\log \alpha$ forms the equivalent problem

\begin{align}
    \begin{aligned}
        \min_{\boldsymbol{p}} \ &\alpha \\
        \text{s.t.} \ &\tilde{\alpha} \leq \log p_i + \sum_{j \in I_i} \log (1 - \frac{p_j}{1 + d_{ij}}), \ \forall \ i \\
        &0 \leq p_i \leq 1, \ \forall \ i.
     \end{aligned}
\end{align}

Because log is a monotonic function, the objective can be rewritten as a maximization over $\tilde{\alpha}$. Recognizing that $\alpha \geq 1$, and therefore $\tilde{\alpha}$ is negative, this is equivalent to minimizing its square.

\begin{align}
    \begin{aligned}
        \min_{\boldsymbol{p}} \ &\frac{1}{2} \tilde{\alpha}^2 \\
        \text{s.t.} \ &\tilde{\alpha} \leq \log p_i + \sum_{j \in I_i} \log (1 - \frac{p_j}{1 + d_{ij}}), \ \forall \ i \\
        &0 \leq p_i \leq 1, \ \forall \ i.
     \end{aligned}
\end{align}

Taking the Lagrangian dual,

\begin{align}
\label{eq:minmaxdual}
    \begin{aligned}
        \min_{\boldsymbol{p}} \ &\frac{1}{2} \tilde{\alpha}^2 + \sum_{i=1}^N \lambda_i \big( \tilde{\alpha} - \log p_i - \sum_{j \in I_i} \log (1 - \frac{p_j}{1 + d_{ij}}) \big) \\
        \text{s.t.} \ &0 \leq p_i \leq 1, \ \forall \ i \\
        &\lambda_i \geq 0, \ \forall \ i.
     \end{aligned}
\end{align}

The primal problem is convex and Slater's condition is satisfied \cite{boyd2004convex}, so the duality gap is once again zero and the maximum of the solution over $\boldsymbol{\lambda}$ is equivalent to the primal problem. Because \eqref{eq:minmaxdual} is convex in $\boldsymbol{p}$ and $\tilde{\alpha}$ and concave in $\boldsymbol{\lambda}$, we take the gradient and set it equal to zero to find the solution

\begin{subnumcases}{}
    \frac{\lambda_i}{p_i} = \sum_{j \in I_i} \frac{\lambda_j}{1 + d_{ji} - p_i} & \\
    \tilde{\alpha} + \sum_{i=1}^N \lambda_i = 0  \\
    \begin{aligned}
        \tilde{\alpha} &= \log p_i + \sum_{j \in I_i} \log (1 - \frac{p_j}{1 + d_{ij}}) \\
        &= - \log h_i
    \end{aligned} & \label{subeq:mmalpha}
\end{subnumcases}

for all $i$. Observe from \eqref{subeq:mmalpha} that $\tilde{\alpha}$ does not depend on $i$, so the min max solution achieves an equal expected AoI for all nodes.

This solution is unique by the same argument as EWSAoI. \eqref{eq:mmopt} is monotonically decreasing in $p_i$ and goes to infinity as $p_i$ goes to $0$. If the left hand side evaluated at $p_i = 1$ is less than $0$, then by the Intermediate Value Theorem there exists a unique solution in the domain $0 \leq p_i < 1$. Otherwise the minimum occurs at $p_i = 1$ by the monotonicity of \eqref{eq:mmopt}. This completes the proof.

\subsection{Proof of Theorem \ref{th:perfinfobound}}
\label{app:perfinfobound}

We begin by showing the lower bound. Because $\theta = 1$, only one packet can be received in each time slot, and the sum of expected throughputs is upper bounded by $1$. In this setting the problem
\begin{align}
\label{eq:relaxedminsum}
    \begin{aligned}
        \min \ &\sum_{i=1}^N h_i(\boldsymbol{r}) \\
        \text{s.t.} \ &\sum_{i=1}^N \tau_i(\boldsymbol{r}) \leq 1.
    \end{aligned}
\end{align}

is a relaxed version of the EWSAoI optimization with symmetric weights, so the solution is a lower bound to \eqref{eq:ewsopt}. Using the inverse relationship between AoI and throughput, the solution to \eqref{eq:relaxedminsum} is $h_i = N$ for all $i$. Therefore,
\begin{equation}
    N^2 \leq \sum_{i=1}^N h_i^S(\boldsymbol{r}),
\end{equation}

and normalizing by $N$ gives the result. Furthermore,

\begin{equation}
	\frac{1}{N^2} \sum_{i = 1}^N h_i^S(\boldsymbol{r}) \leq \frac{1}{N^2} \sum_{i=1}^N h^{MM}(\boldsymbol{r}) = \frac{1}{N} h^{MM}(\boldsymbol{r})
\end{equation}

by the definition of min sum. It remains to prove the upper bound, and we begin by showing it holds asymptotically as $N$ goes to infinity.

As shown previously, the dual of the MMAoI optimization \eqref{eq:minmaxdual} has a duality gap of zero, and maximizing this problem over $\boldsymbol{\lambda}$ is equivalent to the primal MMAoI problem. Let the objective function in \eqref{eq:minmaxdual} be $\tilde{h}^{MM}(\boldsymbol{r})$. Then by the minmax inequality,
\begin{equation}
\label{eq:minmaxineq}
    \max_{\boldsymbol{r}} \max_{\boldsymbol{\lambda}} \min_{\boldsymbol{p}} \tilde{h}^{MM}(\boldsymbol{r}) \leq \min_{\boldsymbol{p}} \max_{\boldsymbol{\lambda}} \max_{\boldsymbol{r}} \tilde{h}^{MM}(\boldsymbol{r}).
\end{equation}

Taking the inner maximization of the right hand side and dropping the terms that are independent of $\boldsymbol{r}$,
\begin{equation}
\label{eq:maxrobj}
    \max_{\boldsymbol{r}} \sum_{i=1}^N \lambda_i \big( -\sum_{j \in I_i} \log (1 - \frac{p_j}{1 + d_{ij}}) \big).
\end{equation}

Defining the function
\begin{equation}
    g(i,j) \triangleq \lambda_i \log (1 - \frac{p_j}{1 + d_{ij}}) + \lambda_j \log (1 - \frac{p_i}{1 + d_{ji}}),
\end{equation}

\eqref{eq:maxrobj} can be rewritten as
\begin{equation}
\label{eq:gsum}
    \min_{\boldsymbol{r}} \sum_{(i,j) : i < j} g(i,j).
\end{equation}

We now show that for each $(i,j)$ pair, $g(i,j)$ is minimized for any vector $\boldsymbol{\lambda}$ when $r_i = r_j$. The function is not convex, but is continuous and differentiable, so a necessary condition for a minimum point is that the gradient be zero or to be a boundary point. Taking the derivative with respect to $r_i$ and setting it equal to zero gives
\begin{equation}
    r_i = 0, \ r_i = r_j \sqrt{\frac{p_i \lambda_j -p_j \lambda_i +p_i p_j \lambda_i}{p_j \lambda_i - p_i \lambda_j + p_i p_j \lambda_j}} = r_j \gamma.
\end{equation}

Adding the other boundary point at $r_i = 1$ yields three candidate solutions in terms of $r_i$, and by symmetry the same three candidates for $r_j$. Noting that $\gamma = 1$ when $r_i = r_j$, the possible combinations in both coordinates are $(x,x), (0,1), (1,0), (x, \gamma x), \text{and} (\gamma x, x)$, where $0 \leq x \leq 1$.

If there exists a vector $\boldsymbol{r}^*$ such that $g(i,j)$ is minimized for all $(i,j)$ pairs, then $\boldsymbol{r}^*$ is the solution to \eqref{eq:gsum}. It remains to prove that such a position vector exists.

Solving for $d_{ij}$ and $d_{ji}$ at the candidate minima,
\begin{equation}
    g(i,j) = \begin{cases}
        \lambda_j \log(1 - \frac{p_i}{2}) + \lambda_i \log(1 - \frac{p_j}{2}), &(x,x) \\
        \lambda_j \log(1 - p_i), &(0,1) \\
	  \lambda_i \log(1 - p_j), &(1,0) \\
	  \lambda_j \log (1 - \frac{p_i}{1 + 1/\gamma^2}) + \lambda_i \log (1 - \frac{p_j}{1 + \gamma^2}), &(x, \gamma x) \\
        \lambda_j \log (1 - \frac{p_i}{1 + \gamma^2}) + \lambda_i \log (1 - \frac{p_j}{1 + 1/\gamma^2}), &(\gamma x, x).
    \end{cases}
\end{equation}

From \eqref{eq:ewsopt},
\begin{equation}
	\frac{\lambda_i}{p_i} = \sum_{j \in I_i} C_j,
\end{equation}

when the resulting $p_i \leq 1$, and where $C_j > 0$ for all $j$ and $\lambda_i > 0$ for all $i$. As $N$ goes to infinity so does $|I_i|$, and therefore $p_i \to 0$ for all $i$.

%
%

Taking the limit of the solutions at $(\gamma x, x)$ and $(x, \gamma x)$,
\begin{multline}
	\lim_{p_i \to 0, p_j \to 0} g(i,j) = \\ \lambda_j \log (1 + \frac{\lambda_i}{\lambda_i + \lambda_j}) + \lambda_i \log(1 + \frac{\lambda_j}{\lambda_j + \lambda_i})  > 0
\end{multline}

for both cases. For the other cases,
\begin{equation}
	\lim_{p_i \to 0, p_j \to 0} g(i,j) = 0,
\end{equation}

so the minima occur at $(x,x), (0,1)$, and $(1,0)$, and setting $\boldsymbol{r}^*$ equal to any vector with equal entries $x$ such that $0 \leq x \leq 1$ is a solution to \eqref{eq:gsum}.

Returning to the upper bound in \eqref{eq:minmaxineq}, the maximization over $\boldsymbol{\lambda}$ does not affect the minimization over $\boldsymbol{p}$ because the problem is symmetric at $\boldsymbol{r}^*$ and therefore all values in the vector $\boldsymbol{\lambda}$ are equal. Setting $\lambda_i = 1$ for all $i$ makes the problem equivalent to the PFAoI problem, and from the solution \eqref{eq:pfopt},
\begin{equation}
    \frac{1}{p_i^*} = \sum_{j \in I_i} \frac{1}{1 + d_{ji} - p_i^*}  = \sum_{j \in I_i} \frac{1}{2 - p_i^*} , \ \forall \ i,
\end{equation}

which yields
\begin{equation}
	p_i^* = \frac{2}{N}, \ \forall \ i.
\end{equation}

Plugging this result into \eqref{eq:ageofi} gives the final upper bound
\begin{align}
\label{eq:minmaxbound}
\begin{aligned}
    \frac{1}{N} h^{MM}(\boldsymbol{r}) &\leq \frac{1}{N p_i^* \prod_{j \in I_i} (1 - \frac{p_j^*}{1 + d_{ij}})} \\
    &= \frac{1}{2 \prod_{j \in I_i} (1 - 1/N)} \\
    &= \frac{1}{2(1 - \frac{1}{N})^{N-1}} \xrightarrow{} \frac{e}{2}
\end{aligned}
\end{align}

for any position vector $\textbf{r}$ as $N$ goes to infinity. This completes the proof.

\bibliographystyle{IEEEtran}
\bibliography{minimizing_aoi_spatial_networks}

\begin{thebibliography}{10}
\providecommand{\url}[1]{#1}
\csname url@samestyle\endcsname
\providecommand{\newblock}{\relax}
\providecommand{\bibinfo}[2]{#2}
\providecommand{\BIBentrySTDinterwordspacing}{\spaceskip=0pt\relax}
\providecommand{\BIBentryALTinterwordstretchfactor}{4}
\providecommand{\BIBentryALTinterwordspacing}{\spaceskip=\fontdimen2\font plus
\BIBentryALTinterwordstretchfactor\fontdimen3\font minus
  \fontdimen4\font\relax}
\providecommand{\BIBforeignlanguage}[2]{{%
\expandafter\ifx\csname l@#1\endcsname\relax
\typeout{** WARNING: IEEEtran.bst: No hyphenation pattern has been}%
\typeout{** loaded for the language `#1'. Using the pattern for}%
\typeout{** the default language instead.}%
\else
\language=\csname l@#1\endcsname
\fi
#2}}
\providecommand{\BIBdecl}{\relax}
\BIBdecl

\bibitem{kaul2012real}
S.~Kaul, R.~Yates, and M.~Gruteser, ``Real-time status: How often should one
  update?'' in \emph{2012 Proceedings IEEE INFOCOM}, 2012, pp. 2731--2735.

\bibitem{sun2017update}
Y.~Sun, E.~Uysal-Biyikoglu, R.~D. Yates, C.~E. Koksal, and N.~B. Shroff,
  ``Update or wait: How to keep your data fresh,'' \emph{IEEE Transactions on
  Information Theory}, vol.~63, no.~11, pp. 7492--7508, 2017.

\bibitem{kadota2018scheduling}
I.~Kadota, A.~Sinha, E.~Uysal-Biyikoglu, R.~Singh, and E.~Modiano, ``Scheduling
  policies for minimizing age of information in broadcast wireless networks,''
  \emph{IEEE/ACM Transactions on Networking}, vol.~26, no.~6, pp. 2637--2650,
  2018.

\bibitem{kadota2019scheduling}
I.~Kadota, A.~Sinha, and E.~Modiano, ``Scheduling algorithms for optimizing age
  of information in wireless networks with throughput constraints,''
  \emph{IEEE/ACM Transactions on Networking}, vol.~27, no.~4, pp. 1359--1372,
  2019.

\bibitem{talak2020optimizing}
R.~Talak, S.~Karaman, and E.~Modiano, ``Optimizing information freshness in
  wireless networks under general interference constraints,'' \emph{IEEE/ACM
  Transactions on Networking}, vol.~28, no.~1, pp. 15--28, 2020.

\bibitem{talak2018distributed}
------, ``Distributed scheduling algorithms for optimizing information
  freshness in wireless networks,'' in \emph{2018 IEEE 19th International
  Workshop on Signal Processing Advances in Wireless Communications
  (SPAWC)}.\hskip 1em plus 0.5em minus 0.4em\relax IEEE, 2018, pp. 1--5.

\bibitem{tripathi2019whittle}
V.~Tripathi and E.~Modiano, ``A whittle index approach to minimizing functions
  of age of information,'' in \emph{2019 57th Annual Allerton Conference on
  Communication, Control, and Computing (Allerton)}.\hskip 1em plus 0.5em minus
  0.4em\relax IEEE, 2019, pp. 1160--1167.

\bibitem{yates2017status}
R.~D. Yates and S.~K. Kaul, ``Status updates over unreliable multiaccess
  channels,'' in \emph{2017 IEEE International Symposium on Information Theory
  (ISIT)}.\hskip 1em plus 0.5em minus 0.4em\relax IEEE, 2017, pp. 331--335.

\bibitem{kadota2021age}
I.~Kadota and E.~Modiano, ``Age of information in random access networks with
  stochastic arrivals,'' in \emph{IEEE INFOCOM 2021 - IEEE Conference on
  Computer Communications}, 2021, pp. 1--10.

\bibitem{maatouk2020age}
A.~Maatouk, M.~Assaad, and A.~Ephremides, ``On the age of information in a csma
  environment,'' \emph{IEEE/ACM Transactions on Networking}, vol.~28, no.~2,
  pp. 818--831, 2020.

\bibitem{chen2022age}
X.~Chen, K.~Gatsis, H.~Hassani, and S.~S. Bidokhti, ``Age of information in
  random access channels,'' \emph{IEEE Transactions on Information Theory},
  2022.

\bibitem{yavascan2021analysis}
O.~T. Yavascan and E.~Uysal, ``Analysis of slotted aloha with an age
  threshold,'' \emph{IEEE Journal on Selected Areas in Communications},
  vol.~39, no.~5, pp. 1456--1470, 2021.

\bibitem{ahmetoglu2022mista}
M.~Ahmetoglu, O.~T. Yavascan, and E.~Uysal, ``Mista: An age-optimized slotted
  aloha protocol,'' \emph{IEEE Internet of Things Journal}, 2022.

\bibitem{bianchi2000performance}
G.~Bianchi, ``Performance analysis of the ieee 802.11 distributed coordination
  function,'' \emph{IEEE Journal on selected areas in communications}, vol.~18,
  no.~3, pp. 535--547, 2000.

\bibitem{celik2009mac}
G.~D. Celik, G.~Zussman, W.~F. Khan, and E.~Modiano, ``Mac for networks with
  multipacket reception capability and spatially distributed nodes,''
  \emph{IEEE Transactions on Mobile Computing}, vol.~9, no.~2, pp. 226--240,
  2009.

\bibitem{baccelli2013adaptive}
F.~Baccelli and C.~Singh, ``Adaptive spatial aloha, fairness and stochastic
  geometry,'' in \emph{2013 11th International Symposium and Workshops on
  Modeling and Optimization in Mobile, Ad Hoc and Wireless Networks (WiOpt)},
  2013, pp. 7--14.

\bibitem{baccelli2014analysis}
F.~Baccelli, B.~B{\l}aszczyszyn, and C.~Singh, ``Analysis of a proportionally
  fair and locally adaptive spatial aloha in poisson networks,'' in \emph{IEEE
  INFOCOM 2014-IEEE Conference on Computer Communications}.\hskip 1em plus
  0.5em minus 0.4em\relax IEEE, 2014, pp. 2544--2552.

\bibitem{yang2020optimizing}
H.~H. Yang, A.~Arafa, T.~Q. Quek, and H.~V. Poor, ``Optimizing information
  freshness in wireless networks: A stochastic geometry approach,'' \emph{IEEE
  Transactions on Mobile Computing}, vol.~20, no.~6, pp. 2269--2280, 2020.

\bibitem{mankar2021spatial}
P.~D. Mankar, M.~A. Abd-Elmagid, and H.~S. Dhillon, ``Spatial distribution of
  the mean peak age of information in wireless networks,'' \emph{IEEE
  Transactions on Wireless Communications}, vol.~20, no.~7, pp. 4465--4479,
  2021.

\bibitem{gallager2013stochastic}
R.~G. Gallager, \emph{Stochastic processes: theory for applications}.\hskip 1em
  plus 0.5em minus 0.4em\relax Cambridge University Press, 2013.

\bibitem{gupta2012throughput}
P.~Gupta and A.~L. Stolyar, ``Throughput region of random-access networks of
  general topology,'' \emph{IEEE Transactions on Information Theory}, vol.~58,
  no.~5, pp. 3016--3022, 2012.

\bibitem{massey1985collision}
J.~Massey and P.~Mathys, ``The collision channel without feedback,'' \emph{IEEE
  Transactions on Information Theory}, vol.~31, no.~2, pp. 192--204, 1985.

\bibitem{boyd2004convex}
S.~Boyd, S.~P. Boyd, and L.~Vandenberghe, \emph{Convex optimization}.\hskip 1em
  plus 0.5em minus 0.4em\relax Cambridge university press, 2004.

\bibitem{kelly1997charging}
F.~Kelly, ``Charging and rate control for elastic traffic,'' \emph{European
  transactions on Telecommunications}, vol.~8, no.~1, pp. 33--37, 1997.

\bibitem{dattorro2010convex}
J.~Dattorro, \emph{Convex optimization \& Euclidean distance geometry}.\hskip
  1em plus 0.5em minus 0.4em\relax Lulu. com, 2010.

\end{thebibliography}

\end{document}